\documentclass[conference,letterpaper]{IEEEtran}

\usepackage[utf8]{inputenc} 
\usepackage[T1]{fontenc}
\usepackage{url}
\usepackage{ifthen}
\usepackage{cite}
\usepackage[cmex10]{amsmath} 

\usepackage{enumerate}
\usepackage{amsmath}
\usepackage{amssymb,latexsym}
\usepackage{amsthm}
\usepackage{color}
\usepackage{cancel}
\usepackage{graphicx}
\usepackage{cite}
\usepackage{changes}
\usepackage{hyperref}
\usepackage{comment}
\usepackage{amscd}
\usepackage{mathtools}

\newtheorem{theorem}{Theorem}[section]

\newtheorem{lemma}[theorem]{Lemma}
\newtheorem{corollary}[theorem]{Corollary}
\newtheorem{definition}[theorem]{Definition}
\newtheorem{proposition}[theorem]{Proposition}

\newtheorem{remark}[theorem]{Remark}

\DeclareMathOperator{\C}{\mathcal{C}}

\newcommand{\fqn}{\mathbb{F}_{q^n}}
\newcommand{\fq}{\mathbb{F}_{q}}
\newcommand{\fqk}{\mathbb{F}_{q^k}}
\newcommand{\fqt}{\mathbb{F}_{q^t}}

\newcommand{\F}{{\mathbb F}}

\newcommand{\lmb}{\lambda}
\newcommand{\la}{\langle}
\newcommand{\ra}{\rangle}

\newcommand{\N}{\mathrm{N}}

\begin{document}
\title{On one-orbit cyclic subspace codes of $\mathcal{G}_q(n,3)$} 


\author{%
  \IEEEauthorblockN{Chiara Castello, Olga Polverino and Ferdinando Zullo}
  \IEEEauthorblockA{Universit\`a degli Studi della Campania ``Luigi Vanvitelli''\\ 
                    Viale Lincoln, 5, I--\,81100 Caserta, Italy\\
                    Email: \{chiara.castello,olga.polverino, ferdinando.zullo\}@unicampania.it}
}


\maketitle


\begin{abstract}
   Subspace codes have recently been  used for error correction in random network coding. In this work, we focus on one-orbit cyclic subspace codes. If $S$ is an $\fq$-subspace of $\fqn$, then the one-orbit cyclic subspace code defined by $S$ is
   \[ \mathrm{Orb}(S)=\{\alpha S \colon \alpha \in \fqn^*\}, \]
   where $\alpha S=\lbrace \alpha s \colon s\in S\rbrace$ for any $\alpha\in \fqn^*$.\\
   Few classification results of subspace codes are known, therefore it is quite natural to initiate a classification of cyclic subspace codes, especially in the light of the recent classification of the isometries for cyclic subspace codes. We consider three-dimensional one-orbit cyclic subspace codes, which are divided into three families: the first one containing only $\mathrm{Orb}(\mathbb{F}_{q^3})$; the second one containing the optimum-distance codes; and the third one whose elements are codes with minimum distance $2$. We study inequivalent codes in the latter two families. 
\end{abstract}

\section{Introduction}

Let $k$ be a non-negative integer with $k \leq n$. The set of all $k$-dimensional $\F_q$-subspaces of $\F_{q^n}$, viewed as an $\F_{q}$-vector space, forms a \textbf{Grassmannian space} over $\F_q$, which is denoted by $\mathcal{G}_{q}(n,k)$. A \textbf{constant dimension subspace code} is a subset $C$ of $\mathcal{G}_{q}(n,k)$ endowed with the metric defined as follows \[d(U,V)=\dim_{\F_q}(U)+\dim_{\F_q}(V)-2\dim_{\F_q}(U \cap V),\]
where $U,V \in C$. This metric is also known as \textbf{subspace metric}.
As usual, we define the \textbf{minimum distance} of $C$ as
\[ d(C)=\min\{ d(U,V) \colon U,V \in C, U\ne V \}. \]
Subspace codes have been recently used for the error correction in random
network coding, see \cite{KoetterK}. 
The first class of subspace codes studied was the one introduced in \cite{Etzion}, known as \textbf{cyclic subspace codes}.
A subspace code $C \subseteq \mathcal{G}_q(n,k)$ is said to be \textbf{cyclic} if $\alpha V\in\C$ for every $\alpha \in \F_{q^n}^*$ and every $V \in C$.
If $C$ coincides with $\mathrm{Orb}(S)$, for some subspace $S$ of $\fqn$, we say that $C$ is a \textbf{one-orbit} cyclic subspace code and $S$ is said to be a \textbf{representative} of the orbit. \\
Let $S$ be an $\fq$-subspace of $\fqn$ of dimension $k$ and let $d=\gcd(n,k)$. Then $|\mathrm{Orb}(S)|=\frac{q^n-1}{q^d-1}$ if and only if $\mathbb{F}_{q^d}$ is the maximum subfield of $\fqn$ of linearity of $S$; see \cite[Theorem 1]{Otal}.
Therefore, every orbit of a subspace $V \in \mathcal{G}_q(n,k)$ defines a cyclic subspace code of size $(q^n-1)/(q^t-1)$, for some $t \mid n $. Let $S$ be a strictly $\fq$-linear subspace of dimension $k$, i.e. $\fq$ is the maximum field on which $S$ is linear. Then the cyclic subspace code defined by $S$ has size $(q^n-1)/(q-1)$. Also, in this case, the maximum value for the minimum distance is at most $2k$ and it is exactly $2k$ if and only if the orbit of $S$ is a $k$-spread of $\F_{q^n}$, i.e. $\mathrm{Orb}(\F_{q^k})$ and $k=1$ because of the linearity assumption. 

In \cite{Trautmann} the authors conjectured the existence of  a cyclic code of size $\frac{q^n-1}{q-1}$ in $\mathcal{G}_q(n,k)$ and minimum distance $2k-2$ for every pair of positive integers $n,k$ such that $1<k\leq n/2$. These codes are also known as \textbf{optimum-distance codes}.

In \cite{BEGR} the authors used subspace polynomials to generate cyclic subspace codes with size $\frac{q^n-1}{q-1}$ and minimum distance $2k-2$, proving that the conjecture is true for any given $k$ and infinitely many values of $n$. This was improved in \cite{Otal}. 
Finally, the conjecture was solved in \cite{Roth} for most of the cases, by making use of Sidon spaces originally introduced in \cite{BSZ2015}, in relation with the linear analogue of Vosper's Theorem (see also \cite{BSZ2018,HouLeungXiang2002}). An $\fq$-subspace $S$ of $\fqn$ is called a \textbf{Sidon space} if $S$ satisfies the following property: for all nonzero $a,b,c,d \in S$ such that $ab=cd$ then 
$\{a \F_q,b \F_q\}=\{c \F_q,d \F_q\}$,
where $e\fq =\{e\lambda \colon \lambda \in \fq\}$.
So, the study of cyclic subspace codes with size $\frac{q^n-1}{q-1}$ and minimum distance $2k-2$ is equivalent to the study of Sidon spaces. Indeed, the following theorem explains the condition $k\leq n/2$ in the aforementioned conjecture. To this aim we need the following notation: if $S$ is an $\fq$-subspace of $\fqn$ then
$ S^2=\langle ab \colon a,b \in S\rangle_{\fq}$, 
i.e. $S^2$ is the $\fq$-subspace spanned by the products of pairs of elements of $S$.
\begin{theorem}\cite[Theorem 18]{BSZ2015} \label{lowerboundSidon}
Let $S\in\mathcal{G}_q(n,k)$ be a Sidon space of dimension $k\geqslant 3$, then
 $
\dim_{\fq}(S^2)\geqslant 2k$.
\end{theorem}
Also, Sidon spaces with the smallest dimension of its square span were used in designing a multivariate public-key cryptosystem (see \cite{raviv2021multivariate}).\\
Apart from some classification results on subspace codes, see for instance \cite{equidcodes,lisa2,lisa}, very few classification results are known for cyclic subspace codes. In light of the classification of the isometries for cyclic subspace codes recently proved by Gluesing-Luerssen and Lehmann in \cite{Heideequiv} (see also \cite{Trautmann2}), it is quite natural to initiate a classification of cyclic subspace codes.

In this paper we will consider $3$-dimensional one-orbit cyclic subspace codes and we will give some classification results.
The possible $3$-dimensional one-orbit cyclic subspace codes $\mathrm{Orb}(S)$ are of the following types:
\begin{itemize}
    \item[i)] $|\mathrm{Orb}(S)|=(q^n-1)/(q^3-1)$ and $d(\mathrm{Orb}(S))=6$;
    \item[ii)] $|\mathrm{Orb}(S)|=(q^n-1)/(q-1)$ and $d(\mathrm{Orb}(S))=4$;
    \item[iii)] $|\mathrm{Orb}(S)|=(q^n-1)/(q-1)$ and $d(\mathrm{Orb}(S))=2$.
\end{itemize}
For Case i), we only have $\mathrm{Orb}(\F_{q^3})$.
The codes of the Family ii) are the optimum-distance codes and are those for which $S$ is a Sidon space; whereas the codes of Family iii) correspond to those subspaces for which there exists at least one $\alpha \in \fqn^*$ such that $\dim_{\fq}(S\cap \alpha S)=2$. 
So, the problem is to determine inequivalent codes in both Family ii) and iii).
To this aim, we first introduce in Section \ref{sec:equivandinv} some new invariants which can be used to distinguish inequivalent classes of codes, based on the square-span of a subspace and the span of a subspace over a larger field. In Section \ref{sec:afirstclass} we give a classification result based on the dimension of the square-span of a representative of the code and we study the equivalence problem for the codes in Family iii). In the last section (cf. Section \ref{sec:optdistcodes}) we investigate the equivalence problem for the codes in Family ii) under the assumption that a representative is contained in the sum of two multiplicative cosets of $\F_{q^3}$.
Some of the technical proofs are in the Appendix to keep this paper short.

\section{Equivalence and invariants}\label{sec:equivandinv}

The study of the equivalence for subspace codes was initiated by Trautmann in \cite{Trautmann2} and the case of cyclic subspace codes has been investigated in \cite{Heideequiv} by Gluesing-Luerssen and Lehmann.
Therefore, motivated by \cite[Theorem 6.2 (a)]{Heideequiv}, we say that two cyclic subspace codes $\mathrm{Orb}(S_1)$ and $\mathrm{Orb}(S_2)$ in $\fqn$ are \textbf{linearly equivalent}  if there exists $i \in \{0,\ldots,n-1\}$ such that 
$\mathrm{Orb}(S_1)=\mathrm{Orb}(S_2^{q^i})$, 
where $S_2^{q^i}=\{ v^{q^i} \colon v \in S_2 \}$. This happens if and only if $S_1=\alpha S_2^{q^i}$, for some $\alpha \in \fqn^*$.
We can replace the action of the $q$-Frobenius maps $x\in \fqn\mapsto x^{q^i}\in \fqn$ with any automorphism $\sigma$ in $\mathrm{Aut}(\fqn)$, since this will still preserve the metric properties of the codes. 
Following \cite{Zullo}, we consider an extension of this definition, where we will denote by $S^{\sigma}=\lbrace \sigma(s)\colon s\in S\rbrace$ with $\sigma\in\mathrm{Aut}(q^n)$.

\begin{definition}
\label{def1}
Let $S_1$ and $S_2$ be two $\fq$-subspaces of $\fqn$. Then we say that $S_1$ and $S_2$ are \textbf{semilinearly equivalent} (or simply \textbf{equivalent}) if the associated codes $\mathrm{Orb}(S_1)$ and $\mathrm{Orb}(S_2)$ are semilinearly equivalent, that is there exist $\sigma \in \mathrm{Aut}(\fqn)$ and $\alpha \in \fqn$ such that $S_1=\alpha S_2^{\sigma}$.
In this case, we will also say that they are equivalent under the action of the pair $(\alpha, \sigma)\in\fqn\rtimes \mathrm{Aut}(\fqn)$.
\end{definition}

In the following we describe some invariants that can be used to establish whether or not two one-orbit cyclic subspace codes are equivalent.

The first invariant that we introduce regards the dimension of the square-span of a subspace. Indeed, the following is easy to check.

\begin{proposition}
Let $S_1,S_2$ be two $\fq$-subspaces of $\fqn$ and suppose that $S_1$ and $S_2$ are equivalent under the action of $(\alpha, \sigma)$, then $S_1^2$ and $S_2^2$ are equivalent under the action of $(\alpha^2, \sigma)$. In particular, $\dim_{\fq}(S_1^2)=\dim_{\fq}(S_2^2)$.
\end{proposition}

\begin{definition}
For any divisor $t$ of $n$ and for any $\fq$-subspace $S$ of $\fqn$, we denote by 
\[
\delta_t(S):=\dim_{\fqt}(\langle S\rangle_{\F_{q^t}})
\]
and 
\[
w_t(S):=\max\lbrace \dim_{\fqt}(S')\colon S'\in\mathcal{S}'\rbrace
\]
where \[
\mathcal{S}'=\lbrace S'\colon S' \text{ is }\fq\text{-subspace of }S\text{ and }\fqt\text{-subspace of }\fqn\rbrace.\]
\end{definition}

Note that the integer $\delta_t(S)$ has been introduced in \cite[Definition 4.5]{Heideequiv}. Clearly, if $\dim_{\fq}(S)>0$, then
\[
1\leq \delta_t(S)\leq \min\left\{ \frac{n}{t},\dim_{\fq}(S)\right\}
\]
and
\[
0\leq w_t(S)\leq \frac{\dim_{\fq}(S)}{t}.
\]
It is easy to see that these two integers are invariant under semilinear equivalence.
\begin{proposition}\label{prop:Fqt}
\label{prop:deltaewtinvariants}
Let $S_1,S_2$ be two $\fq$-subspaces of $\fqn$ and suppose that $S_1$ and $S_2$ are equivalent under the action of $(\alpha, \sigma)$.
If $t \mid n$, then $\langle S_1 \rangle_{\F_{q^t}}$ and $\langle S_2 \rangle_{\F_{q^t}}$ are equivalent under the action of $(\alpha, \sigma)$. In particular, $\delta_t(S_1)=\delta_t(S_2)$ and $w_t(S_1)=w_t(S_2)$.
\end{proposition}

\section{A classification result}\label{sec:afirstclass}

We now present a classification of three dimensional $\fq$-subspaces of $\fqn$, based on the equivalence given in Definition \ref{def1}, yielding a classification of three dimensional one-orbit cyclic subspace codes, by making use of the following lemma (an extension of \cite[Lemma 4]{BSZ2015}).

\begin{lemma}\cite[Lemma 3.1]{NPSZminsize}\label{lemma:power}
Let $S$ be an $\fq$-subspace of $\fqn$ of dimension $k\geq2$ and let $\lmb \in \fqn\setminus\fq$. Let $t=\dim_{\fq}(\fq(\lmb))$, where $\fq(\lmb)$ denotes the field extension of $\fq$ generated by $\lmb$.
\begin{itemize}
    \item [(a)] If $\dim_{\fq}(S\cap \lmb S)=k$, then $S$ is an $\fq(\lmb)$-subspace.
    \item [(b)] Suppose that $\dim_{\fq}(S\cap \lmb S)=k-1$ and $t\geq k$. Then $S=\mu \langle 1,\lmb,\ldots,\lmb^{k-1}\rangle_{\fq}$, for some $\mu \in \fqn^*$ and $t \neq k$.
    \item[(c)] Suppose that $\dim_{\fq}(S\cap \lmb S)=k-1$ and $t\leq k-1$. Write $k=t\ell+m$ with $m<t$, then $m>0$ and $S=\overline{S}\oplus \mu\langle 1,\lmb,\ldots,\lmb^{m-1}\rangle_{\fq}$, where $\overline{S}$ is an $\F_{q^t}$-subspace of dimension $\ell$, $\mu \in \fqn^*$ and $\mu \F_{q^t} \cap \overline{S}=\{0\}$.
    In particular, $t$ is a proper divisor of $n$.
\end{itemize}
\end{lemma}

We are now ready to prove our first classification result, which is mainly based on the first invariant introduced in Section \ref{sec:equivandinv}.

\begin{theorem}\label{thm:class3dim}
Let $S$ be an $\fq$-subspace of $\fqn$ such that $\dim_{\fq}(S)=3$. Then
\begin{itemize}
    \item [1)] $\dim_{\fq}(S^2)=3$ if and only if $S$ is a multiplicative coset of $\mathbb{F}_{q^3}$, i.e. $3\mid n$ and $S=\mu \F_{q^3}$ for some $\mu \in \fqn$; \vspace{0.3cm}
    \item[2)] $\dim_{\fq}(S^2)=4$ if and only if one of the following holds:
        \begin{itemize}
            \item [2.1)] $S=\mu\la 1,\lmb,\lmb^2\ra_{\fq}$ for some $\mu,\lmb\in\fqn\setminus\fq$, $4\mid n$ and $\delta_4(S)=1$.
            \item[2.2)] $S=\omega\mathbb{F}_{q^2}+\la \mu\ra_{\fq}$ for some $\mu,\omega\in\fqn$ such that $\mu\notin\omega\mathbb{F}_{q^2}$, $4\mid n$, $\delta_4(S)=1$ and $w_2(S)=1$.
        \end{itemize}
        \vspace{0.3cm}
    \item[3)] $\dim_{\fq}(S^2)=5$ if and only if one the following holds:
        \begin{itemize}
            \item [3.1)] $S=\mu\la 1,\lmb,\lmb^2\ra_{\fq}$ for some $\mu,\lmb\in\fqn\setminus\fq$ such that $\dim_{\fq}(\fq(\lmb))>4$; also, if $2\mid n$ then $\delta_2(S)=3$ and $w_2(S)=0$.
            \item[3.2)] $S=\omega\mathbb{F}_{q^2}+\la \mu\ra_{\fq}$ for some $\mu,\omega\in\fqn$ such that $\mu\notin\omega\mathbb{F}_{q^4}$, $\delta_2(S)=2$ and $w_2(S)=1$.\vspace{0.3cm} 
        \end{itemize}
    \item[4)] $\dim_{\fq}(S^2)=6$ if and only if $S$ is a Sidon space. 
\end{itemize}
\end{theorem}
\begin{proof}
Since $\dim_{\fq}(S)=3$ then $3\leqslant \dim_{\fq} (S^2)\leqslant 6$, so we have four cases to analyze and we split the analysis according to the dimension of the square-span of $S$.\\
\textbf{Case 1)}
Suppose that $\dim_{\fq}(S^2)=3=\dim_{\fq}(S)$ then, without loss of generality we may assume that $1\in S$. Then $S^2=S$ and so for any $\lmb\in S\setminus\fq$, we get $S=\lmb S$, i.e. $S$ is $\fq(\lmb)$-subspace of $\fqn$ and since $\dim_{\fq}(S)=3$ then $\fq(\lmb)=\mathbb{F}_{q^3}$ and $3\mid n$. \\
\textbf{Case 2)} Suppose that $\dim_{\fq}(S^2)=4$, then $S$ is not a Sidon space by Theorem \ref{lowerboundSidon}, i.e. there exists $\lmb\in\fqn\setminus\fq$ such that
\[
\dim_{\fq}(S\cap \lmb S)>1,
\]
and hence
$2\leqslant\dim_{\fq}(S\cap \lmb S)\leqslant 3$. If $\dim_{\fq}(S \cap \lmb S)=3=\dim_{\fq}(S)$ then $S=\lmb S$, which implies $\fq(\lmb)=\mathbb{F}_{q^3}$ and $\dim_{\fq}( S^2)=3$, a contradiction to $\dim_{\fq}(S^2)=4$. Therefore $\dim_{\fq}(S\cap \lmb S)=2=\dim_{\fq}(S)-1$. Lemma \ref{lemma:power}  implies that one of the following cases occurs
\begin{itemize}
    \item[ 2.1)] $S=\mu\la 1,\lmb,\lmb^2\ra_{\fq}$ for some $\mu\in\fqn^{*}\setminus\fq$ and $\fq(\lmb)\neq\mathbb{F}_{q^2}$;
    \item[ 2.2)] $S=\omega\mathbb{F}_{q^2}+\la \mu\ra_{\fq}$ for some $\mu\in\fqn\setminus\omega\mathbb{F}_{q^2}$ and $\fq(\lmb)=\mathbb{F}_{q^2}$.
\end{itemize}
If $S=\mu\la 1,\lmb,\lmb^2\ra_{\fq}$, then $S^2=\mu^2\la 1,\lmb,\lmb^2,\lmb^3,\lmb^4\ra_{\fq}$ and since $\dim_{\fq}(S^2)=4$, the elements $1,\lmb,\lmb^2,\lmb^3,\lmb^4$ are $\fq$-linearly dependent, i.e. $\lmb$ is a root of a non-zero polynomial of degree less than or equal to $4$ over $\fq$. In particular, if the minimal polynomial of $\lmb$ over $\fq$ has degree strictly less than $4$, then this would give a contradiction to $\dim_{\fq}( S^2)=4$. Therefore the minimal polynomial of $\lmb$  over $\fq$ has degree $4$, which implies that $\fq(\lmb)=\mathbb{F}_{q^4}$, so $4|n$ and $S\subseteq\mu\mathbb{F}_{q^4}$, i.e. $\delta_4(S)=1$.\\
If $S=\omega\mathbb{F}_{q^2}+\la \mu\ra_{\fq}$, then $S^2=\omega^2\mathbb{F}_{q^2}+\omega\mu\mathbb{F}_{q^2}+\la\mu^2\ra_{\fq}$. Let observe that if $\omega^2\mathbb{F}_{q^2}=\omega\mu\mathbb{F}_{q^2}$ then $\frac{\mu}{\omega}\in\mathbb{F}_{q^2}$, i.e. $\mu\in\omega\mathbb{F}_{q^2}$, a contradiction. Therefore $\mu^2\in\omega^2\mathbb{F}_{q^2}+\omega\mu\mathbb{F}_{q^2}$, i.e. there exist $\alpha,\beta\in\mathbb{F}_{q^2}$ such that $\mu^2=\alpha\omega^2+\beta\omega\mu$. This implies that $\frac{\omega}{\mu}$ is a root of the polynomial $\alpha x^2+\beta x -1=0$ whose coefficients are in $\mathbb{F}_{q^2}$. Then $\frac{\omega}{\mu}\in\mathbb{F}_{q^4}\setminus \F_{q^2}$. So, we have that there exists $\rho\in\mathbb{F}_{q^4}\setminus\mathbb{F}_{q^2}$ such that $\mu=\rho\omega$. Then
\[
S=\omega\mathbb{F}_{q^2}+\omega\la \rho\ra_{\fq}=\omega(\mathbb{F}_{q^2}+\la\rho\ra_{\fq})\subset \omega\mathbb{F}_{q^4}
\]
and so $4\mid n$.\\
\textbf{Case 3)} Suppose that $\dim_{\fq}(S^2)=5$, then $S$ is not a Sidon space by Theorem \ref{lowerboundSidon}, and so arguing as before, there exists $\lmb\in\fqn\setminus\fq$ such that $\dim_{\fq}(S\cap \lmb S)=2=\dim_{\fq}(S)-1$. Then Lemma \ref{lemma:power} and $\dim_{\fq}(S^2)=5$ implies that $S$ has one of the following forms:
\begin{itemize}
    \item [3.1)] $S=\mu\la 1,\lmb,\lmb^2\ra_{\fq}$ for some $\mu\in\fqn^{*}\setminus\fq$ and $\dim_{\fq}\fq(\lmb)> 4$;
    \item[3.2)] $S=\omega\mathbb{F}_{q^2}+\la \mu\ra_{\fq}$ for some $\mu\in\fqn\setminus\omega\mathbb{F}_{q^2}$ and $\fq(\lmb)=\mathbb{F}_{q^2}$ and $\mu^2\notin \omega^2\mathbb{F}_{q^2}+\omega\mu\mathbb{F}_{q^2}$.
\end{itemize}
In particular, in Case 3.2), $\mu^2\notin\omega^2\mathbb{F}_{q^2}+\omega\mu\mathbb{F}_{q^2}$ because otherwise $\frac{\mu}{\omega} $ would be root of a non zero polynomial of degree 2 over $\mathbb{F}_{q^2}$, hence $\frac{\mu}{\omega}\in\mathbb{F}_{q^4}$. In this case $\omega S\subseteq \F_{q^4}$ and so $\dim_{\fq}(S^2)\leq 4$, a contradiction. \\
Moreover, in Case 3.1), since $\lmb\notin\mathbb{F}_{q^2}$, if $\delta_2(S)=2$, then there exist $a,b\in\mathbb{F}_{q^2}$ such that $\lmb^2=a+b\lmb$, i.e. $\mathbb{F}_{q^2}(\lmb)=\mathbb{F}_{q^4}$, a contradiction. Hence $\delta_2(S)=3$. If there exists $\xi\in\fqn$ such that $\xi\mathbb{F}_{q^2}\subseteq S$, then $S\subseteq \xi\mathbb{F}_{q^2}\oplus\mu\mathbb{F}_{q^2}$, where $\mu\in S\setminus \xi\mathbb{F}_{q^2}$, and since $\delta_2(S)=3$, we have a contradiction. Thus $w_2(S)=0$.\\ In Case 3.2) it is clear that $\delta_2(S)=2$ and $w_2(S)=1$.\\
\textbf{Case 4)} If $\dim_{\fq}(S^2)=6$, since $\dim_{\fq}(S)=3$, then $S^2$ has its maximum possible dimension and by \cite[Lemma 20]{Roth} it follows that $S$ is a Sidon space. Conversely, if $S$ is a Sidon space, then Theorem \ref{lowerboundSidon} implies the assertion.
\end{proof}

Cases 2.1 and 2.2 give rise to equivalent examples. To see this, we need the following well-known lemma (which follows by \cite[Theorem 2.24]{lidl_finite_1997}). 

\begin{lemma}
\label{lem:hyperplanes}
Let $H_1,H_2$ be two $\fq$-subspaces of $\fqn$ such that $\dim_{\fq}(H_1)=\dim_{\fq}(H_2)=n-1$, then there exists $\xi\in\fqn^*$ such that $H_2=\xi H_1$.
\end{lemma}

\begin{proposition}
\label{prop:allCase2equivalent}
Let $S_1$ and $S_2$ be two three-dimensional $\fq$-subspaces of $\fqn$ such that
\[ S_1=\mu\la 1,\lmb,\lmb^2\ra_{\fq}\subseteq \mu\mathbb{F}_{q^4}\,\,\,\text{and}\,\,\,S_2=\omega\mathbb{F}_{q^2}+\la\eta\ra_{\fq}\subseteq \omega\mathbb{F}_{q^4},  \]
for some $\mu,\lmb\in\fqn\setminus\fq$ and $\eta\in\fqn\setminus\omega\mathbb{F}_{q^2}$. Then $S_1$ and $S_2$ are equivalent.
\end{proposition}
\begin{proof}
By Lemma \ref{lem:hyperplanes} there exists $\xi\in\fqn^*$ such that $\mu^{-1}S_1=\xi(\omega^{-1} S_2)$, that is $S_1=\xi\mu\omega^{-1} S_2$, i.e. $S_1$ and $S_2$ are equivalent.
\end{proof}

\begin{remark}
Proposition \ref{prop:allCase2equivalent} shows that all the subspaces in Case 2) are equivalent to a subspace of Case 2.1), i.e. they all admit a polynomial basis. Moreover, by Theorem \ref{thm:class3dim} and Proposition \ref{prop:deltaewtinvariants} we get that three dimensional subspaces in $\fqn$ as in Cases 3.1) and 3.2) are not equivalent under the action of semilinear equivalence.
\end{remark}

We now discuss the equivalence among the codes as in Case 3.1) and later those of in Case 3.2).

\begin{theorem}\label{thm:polcaseequiv}
Let $S,T$ be two $\fq$-subspaces of $\fqn$ of dimension $3$ such that 
\[
S=\langle 1,\lmb,\lmb^2\rangle_{\fq}\,\,\,\text{ and }\,\,\, T=\langle 1,\mu,\mu^2\rangle_{\fq}
\]
for some $\lmb,\mu\in\fqn\setminus\fq$ such that $\dim_{\fq}(\fq(\lmb))>4$ and $\dim_{\fq}(\fq(\mu))>4$. Then $S$ and $T$ are equivalent under the action of $(\xi,\sigma)\in \fqn^*\rtimes \mathrm{Aut}(\fqn)$ if and only if $\lmb^{\sigma}=\frac{\alpha_0+\alpha_1\mu}{\beta_0+\beta_1\mu}$ with $(\alpha_1,\beta_1)\neq(0,0)$.
\end{theorem}

\begin{theorem}\label{thm:inequiv123}
Let $S,T$ be two $\fq$-subspaces of $\fqn$ of dimension $3$ such that 
\[
S=\mathbb{F}_{q^2}+\langle\mu\rangle_{\fq}\,\,\,\text{ and }\,\,\, T=\mathbb{F}_{q^2}+\langle \eta\rangle_{\fq}
\]
for some $\mu,\eta\in\fqn\setminus\mathbb{F}_{q^2}$. Then $S$ and $T$ are equivalent under the action of $(\xi,\sigma)\in \fqn^*\rtimes \mathrm{Aut}(\fqn)$ if and only if $\xi\in\mathbb{F}_{q^2}$ and  $\eta=a+\xi\mu^{\sigma}b$ where $a\in\mathbb{F}_{q^2}$ and $b\in\fq^*$.
\end{theorem}

In terms of codes we obtain the following result, as a corollary of the results of this section.

\begin{corollary}
\label{cor:orbitclass}
Let $C$ be a one-orbit cyclic subspace code with dimension three in $\fqn$. Then $C$ is equivalent to $\mathrm{Orb}(S)$ where $S$ satisfies one of the following conditions
\begin{itemize}
    \item[I)] $S=\F_{q^3}$, $d(C)=6$;
    \item [II)]$\dim_{\fq}(S^2)=4$, $S=\la 1,\lmb,\lmb^2\ra_{\fq}$ for some $\lmb \in \F_{q^4}\setminus \F_{q^2}$, $d(C)=2$, $\delta_4(S)=1$ and $w_2(S)=1$;
    \item [III)]$\dim_{\fq}(S^2)=5$, $S=\la 1,\lmb,\lmb^2\ra_{\fq}$ for some $\lmb \in \F_{q^n}\setminus \F_{q^4}$, $d(C)=2$, if $n$ is even $\delta_2(S)=3$ and $w_2(S)=0$;
    \item [IV)]$\dim_{\fq}(S^2)=5$, $S=\F_{q^2}+\la \mu\ra_{\fq}$ for some $\mu \in \F_{q^n}\setminus \F_{q^4}$, $d(C)=2$, $\delta_2(S)=2$ and $w_2(S)=1$;
    \item [V)]$S$ is a Sidon space, $d(C)=4$.
\end{itemize}
\end{corollary}

\begin{remark}
By using Corollary \ref{cor:orbitclass} together with Lemma \ref{lem:hyperplanes} and Theorems \ref{thm:class3dim}, \ref{thm:polcaseequiv} and \ref{thm:inequiv123}, we can get information about the number of inequivalent codes. The
subspaces of Family II) of Corollary \ref{cor:orbitclass}, up to equivalence, generate only one orbit. Whereas, when $n$ is odd, it can be proved that the number $t$ of inequivalent orbits of those having minimum distance $2$ is bounded as follows
\[
\frac{q^{n-1}-1}{nh(q^2-1)}\leq t \leq \frac{q^{n-1}-1}{q^2-1},
\]
where $q=p^h$, $p$ prime, $h\in\mathbb{N}$.
When $h=1$, $n$ is a prime number such that $n>p+1$, then $t$ reaches the above lower bound. For small values of $q$ and $n$ this number has been computed also in \cite{Heideequiv,Heideequiv2}. 
In the next section, we will investigate the equivalence issue for subspaces of Family V).
\end{remark}

\section{Optimum-distance codes}\label{sec:optdistcodes}

In this section we will deal with three-dimensional one-orbit cyclic subspace codes of size $\frac{q^n-1}{q-1}$ having minimum distance $4$, under the assumption that  $\delta_3(S)=2$. 
We will restrict our study to this family, since these are rare objects with respect to those having $\delta_3(S)=3$. 
We start by proving that these subspaces/codes admit an interesting polynomial description via linearized polynomials. Recall that a \textbf{$q$-polynomial}/\textbf{linearized polynomial} over $\fqn$ is a polynomial of the form
$\sum_{i=0}^t a_ix^{q^i} \in \fqn[x]$, 
for some $t \in \mathbb{N}_0$, and denote by $\mathcal{L}_{n,q}$ the set of $q$-polynomials over $\fqn$.

In the following we prove that a subspace $S$ with $\delta_3(S)=2$ can be represented via a linearized polynomial.

\begin{proposition}\label{prop:1polstep}
Let $S$ be an $\fq$-subspace of $\fqn$ of dimension $3$ over $\fq$, with $n=3s$ and $s\geqslant 2$, such that $\delta_3(S)=2$. Then there exists $\xi\in\fqn^*$ such that $S\cap\xi\mathbb{F}_{q^3}=\lbrace 0\rbrace$ and $\xi\mathbb{F}_{q^3}\subseteq\langle S\rangle_{\mathbb{F}_{q^3}}$. 
\end{proposition}

\begin{remark}
\label{rmk:Sdisjointlmbfq3}
By the previous proposition, if $\delta_3(S)=2$ then there exist $\xi,\rho\in\fqn$ such that $\xi\mathbb{F}_{q^3}\cap S=\lbrace 0\rbrace$, $\xi\mathbb{F}_{q^3}\subseteq \langle S\rangle_{\F_{q^3}}$ and $\xi\mathbb{F}_{q^3}+\rho\mathbb{F}_{q^3}=\langle S\rangle_{\mathbb{F}_{q^3}}$. This means that, if $\delta_3(S)=2$, without loss of generality, then we may assume that $S\subseteq \xi\mathbb{F}_{q^3}+\rho\mathbb{F}_{q^3}$ with $\xi,\rho\in\fqn$ such that $S\cap \xi\mathbb{F}_{q^3}=\lbrace 0\rbrace$ and $\frac{\lmb}{\rho}\notin\F_{q^3}$.
\end{remark}

\begin{proposition}
\label{prop:SdefinesFlinearmap} 
Let $S$ be an $\fq$-subspace of $\fqn$, with $n=3s$ and $s\geqslant 2$, such that $\dim_{\fq}(S)=3$ and $ \la S\ra_{\F_{q^3}}=\lmb\mathbb{F}_{q^3}+\rho\mathbb{F}_{q^3}$ with $\frac{\lmb}{\rho}\notin\F_{q^3}$ and $S\cap\lmb \mathbb{F}_{q^3}=\lbrace 0\rbrace$. Then there exists a $q$-polynomial $f\in\mathcal{L}_{3,q}$ such that
$S=\lbrace \rho u+\lmb f(u):u\in\mathbb{F}_{q^3}\rbrace$.
\end{proposition}

By the above proposition we have that
\[
\rho^{-1}S=\lbrace u+\rho^{-1}\lmb f(u):u\in\mathbb{F}_{q^3}\rbrace.
\]
Therefore, up to multiplication by a scalar in $\fqn$, we have that we may represent an $\fq$-subspace $S$ of $\fqn$, with $n=3s$ and $s\geqslant 2$, such that $\dim _{\fq}(S)=3$ and $\delta_3(S)=2$ by a linearized polynomial over $\mathbb{F}_{q^3}$ defined as in Proposition \ref{prop:SdefinesFlinearmap}, i.e.
\[
S=\lbrace u+\gamma f(u):u\in\mathbb{F}_{q^3}\rbrace=V_{f,\xi}.
\]

So, up to equivalence, we may assume that the subspaces we want to study are of the form $V_{f,\gamma}\subseteq\fqn$.
The equivalence among subspaces of the form $V_{f,\gamma}$ has been studied in \cite{CPSZSidon}, in a more general setting.

\begin{theorem}\cite[Theorem 6.2]{CPSZSidon}\label{thm:equiv}
Let $k$ and $n$ be two positive integers such that $k \mid n$.
Let $U,W$ be two $m$-dimensional $\fq$-subspaces of $\mathbb{F}_{q^k}^2$. Consider 
\[ V_{U,\gamma}=\lbrace u+u'\gamma: (u,u')\in U\rbrace\]
and
\[V_{W,\xi}=\lbrace w+w'\xi: (w,w')\in W\rbrace,\]
where $\gamma,\xi\in\fqn$ are such that $\lbrace 1,\gamma\rbrace$ and $\lbrace 1,\xi\rbrace$ are $\fqk$-linearly independent and $\delta_k(V_{U,\gamma})=\delta_k(V_{W,\xi})=2$. Then $V_{U,\gamma}$ and $V_{W,\xi}$ are equivalent under the action of $(\lmb,\sigma)\in\fqn^*\rtimes \mathrm{Aut}(\fqn)$ if and only if there exists $A=\left(\begin{aligned}
    \begin{matrix}
        c & d \\
        a & b 
    \end{matrix}
\end{aligned}\right) \in \mathrm{GL}(2,\fqk)$ such that $\xi=\frac{a+b\gamma^\sigma}{c+d\gamma^\sigma}$, $\lmb=\frac{1}{c+d\gamma^\sigma}$ and $U^{\sigma}=\{w A \colon w \in W\}=W \cdot A$.
\end{theorem}

Clearly, if $U_f=\{ (u,f(u)) \colon u \in \F_{q^k} \}$, which is an $\fq$-subspace of $\fqk^2$, then $V_{f,\gamma}=V_{U_f,\gamma}$.

We are now ready to provide a classification result for three dimensional subspaces $S$ of $\fqn$ such that $\delta_3(S)=2$. We recall that $\mathrm{Tr}_{\F_{q^k}/\fq}(x)=x+x^q+x^{q^2}+\dots+x^{q^{k-1}}$ for $x\in\fqk$.

\begin{theorem}\label{th:classxqTr}
Let $n$ be a multiple of $3$ and let $S$ be an $\fq$-subspace of $\fqn$ with dimension $3$ and $\delta_3(S)=2$. Then there exists $\gamma \in \fqn\setminus \F_{q^3}$ such that $S$ is equivalent  either to $V_{x^q,\gamma}$ or to
$ V_{\mathrm{Tr}_{\F_{q^3}/\fq}(x),\gamma}$.
\end{theorem}
\begin{proof}
Up to equivalence, because of the above results, we can assume that
\[ S=V_{f,\gamma}, \]
for some $f \in \mathcal{L}_{3,q}$ and $\gamma \in \fqn\setminus \F_{q^3}$. Let
$U_f\subseteq \F_{q^3}^2$.
By \cite{LavVdV2} and \cite{FSz} (see also \cite{CsMPclass}), $U_f$ is $\mathrm{\Gamma L}(2,q^3)$-equivalent either to $U_{x^q}$ or to $U_{\mathrm{Tr}_{\F_{q^3}/\fq}(x)}$, i.e. there exist $A=\begin{pmatrix} c&d \\ a & b\end{pmatrix} \in \mathrm{GL}(2,q^3)$ and $\sigma \in \mathrm{Aut}(\F_{q^3})$ such that
$U_f^\sigma =U_{x^q}A \,\,\,\text{or}\,\,\, U_f^\sigma =U_{\mathrm{Tr}_{\F_{q^3}/\fq}(x)}A$. 
By Theorem \ref{thm:equiv} $S$ is equivalent either to $V_{x^q,\gamma'}$ or to $V_{\mathrm{Tr}_{\F_{q^3}/\fq}(x),\gamma'}$, where $\gamma'=\frac{a+b\gamma^\sigma}{c+d\gamma^\sigma}$. 
\end{proof}

We can characterize the values of $\gamma$ for which $V_{x^q,\gamma}$ and $V_{\mathrm{Tr}_{\F_{q^3}/\fq}(x),\gamma}$ are Sidon spaces.
The former one has been already characterized in \cite[Theorems 12 and 16]{Roth} and \cite[Theorem 4.5]{CPSZSidon}, which in our case reads as follows.

\begin{theorem}
Let $\gamma\in\mathbb{F}_{q^{n}}\setminus\F_{q^3}$. 
If $n>6$ then $V_{x^{q},\gamma}$ is a Sidon space for any $\gamma \in \fqn \setminus \F_{q^3}$.
If $n=6$, then $V_{x^{q},\gamma}$ is a Sidon space if and only if $\mathrm{N}_{\F_{q^6}/\fq}(\gamma)=\gamma^{\frac{q^6-1}{q-1}}\neq 1$. 
\end{theorem}

For the case of the trace function we have the following characterization.

\begin{theorem}\label{thm:crucialcharactSidon}
Let $\gamma\in\mathbb{F}_{q^{n}}\setminus\F_{q^3}$. If $n>6$ then $V_{\mathrm{Tr}_{\F_{q^3}/\fq}(x),\gamma}$ is a Sidon space for any $\gamma \in \fqn \setminus \F_{q^3}$.
If $n=6$, then $V_{\mathrm{Tr}_{\F_{q^3}/\fq}(x),\gamma}$ is a Sidon space if and only if $\mathrm{Tr}_{\F_{q^6}/\fq}(\gamma)\ne -2$. 
\end{theorem}
\begin{proof}
Note that $V_{\mathrm{Tr}_{\F_{q^3}/\fq}(x),\gamma}$ is equivalent to $V_{x^q+x^{q^2},\gamma}$,
via the matrix $\begin{pmatrix} 1 & -1\\ 0 & 1 \end{pmatrix}$.
If $n>6$ the assertion follows by \cite[Proposition 4.8]{CPSZSidon}.
So, assume that $n=6$ and let $V=V_{\mathrm{Tr}_{\F_{q^3}/\fq}(x),\gamma}$.
Since $\dim_{\fq}(V)=3$ and $\dim_{\fq}(\ker(\mathrm{Tr}_{\F_{q^3}/\fq}(x)))=2$ we have that 
\[V=\la u_1, u_2, u_3+\gamma\ra_{\fq},\]
for some $u_1,u_2$ and $u_3$ in $\F_{q^3}$ such that $\mathrm{Tr}_{\F_{q^3}/\fq}(u_1)=\mathrm{Tr}_{\F_{q^3}/\fq}(u_2)=0$ and $\mathrm{Tr}_{\F_{q^3}/\fq}(u_3)=1$.\\
Let consider 
\[
V^2=\la u_1^2, u_1u_2, u_1u_3+u_1\gamma, u_2^2, u_2u_3+u_2\gamma, u_3^2+2u_3\gamma+\gamma^2\ra_{\fq}.
\]
Let observe that $u_1^2,u_1u_2,u_2^2$ are $\fq$-linearly independent because otherwise there would exist $\alpha,\beta\in\fq$ such that 
\[
\frac{u_1^2}{u_2^2}=\alpha+\beta\frac{u_1}{u_2},
\]
i.e. $\frac{u_1}{u_2}$ is a root of a polynomial of degree 2 over $\fq$, and this is not possible since $\frac{u_1}{u_2}\in\mathbb{F}_{q^3}\setminus \fq$. Hence $
\mathbb{F}_{q^3}=\la u_1^2,u_1u_2,u_2^2\ra_{\fq}$
and then
\[
V^2=\F_{q^3}+\gamma\langle u_1,u_2,2u_3+\gamma\rangle_{\fq}.
\]
Also $\dim_{\fq}(\mathbb{F}_{q^3}+\la u_1\gamma\ra_{\fq})=4$
since $\gamma\in\mathbb{F}_{q^6}\setminus\mathbb{F}_{q^3}$. Moreover $\dim_{\fq}(\mathbb{F}_{q^3}+\la u_1\gamma,u_2\gamma\ra_{\fq})=5$.
Indeed, if there exist $\alpha\in\fq$ and $\beta\in\mathbb{F}_{q^3}$ such that
$u_2\gamma=\alpha u_1\gamma+\beta$
then
$(u_2-\alpha u_1)\gamma=\beta$
and since $u_2-\alpha u_1$ cannot be zero (because $u_1,u_2$ are $\fq$-linearly independent) this would imply that $\gamma\in\mathbb{F}_{q^3}$, a contradiction. Finally, since $\gamma^2=A+B\gamma$, where $A,B\in\mathbb{F}_{q^3}$ and $B=\mathrm{Tr}_{\F_{q^6}/\F_{q^3}}(\gamma)$ and $A=-\mathrm{N}_{\F_{q^6}/\F_{q^3}}(\gamma)$, we have
\[
2u_3\gamma+\gamma^2=(2u_3+B)\gamma+A
\]
and hence $V^2=\mathbb{F}_{q^3}+\gamma\la u_1,u_2,2u_3+B\ra_{\fq}$.
Note that $\dim_{\fq}(\ker(\mathrm{Tr}_{\F_{q^3}/\fq}(x)))=2$ and $\ker(\mathrm{Tr}_{\F_{q^3}/\fq}(x))=\la u_1,u_2\ra_{\fq}$ then $\dim_{\fq}(V^2)=6$ if and only if $2u_3+B\notin \ker(\mathrm{Tr}_{\F_{q^3}/\fq}(x))$, i.e. $
\mathrm{Tr}_{\F_{q^3}/\fq}(2u_3+B)=2\mathrm{Tr}_{\F_{q^3}/\fq}(u_3)+\mathrm{Tr}_{\F_{q^3}/\fq}(B)=2+\mathrm{Tr}_{\F_{q^3}/\fq}(B)\neq 0,$
which happens if and only if $\mathrm{Tr}_{\F_{q^3}/\fq}(B)=\mathrm{Tr}_{\F_{q^3}/\fq}(\mathrm{Tr}_{\F_{q^6}/\F_{q^3}}(\gamma))=\mathrm{Tr}_{\F_{q^6}/\fq}(\gamma)\neq -2$.
\end{proof}

For more details on the subspace associated with the trace function see \cite{Castello}. 
 
The above result also gives a classification of three dimensional Sidon spaces $S$ with $\delta_3(S)=2$.

\begin{corollary}
Let $n$ be a multiple of three and let $S$ be a Sidon space in $\fqn$ with dimension $3$ and $\delta_3(S)=2$. 
If $n>6$ then $S$ is equivalent to one of the following:
\begin{itemize}
    \item $S_{x^q,\gamma}$;
    \item $S_{\mathrm{Tr}_{\F_{q^3}/\fq}(x),\gamma}$,
\end{itemize}
for some $\gamma \in \fqn\setminus \F_{q^3}$.
If $n=6$ then $S$ is equivalent to one of the following:
\begin{itemize}
    \item $S_{x^q,\gamma}$, for some $\gamma \in \F_{q^6}$ such that $\mathrm{N}_{\F_{q^6}/\fq}(\gamma)\ne 1$;
    \item $S_{\mathrm{Tr}_{\F_{q^3}/\fq}(x),\gamma}$, for some $\gamma \in \fqn$ such that $\mathrm{Tr}_{\F_{q^6}/\fq}(\gamma)\ne -2$.
\end{itemize}
\end{corollary}

We can now use Proposition \ref{prop:weightinvariace} to show that the two families found in Theorem \ref{th:classxqTr}, up to equivalence, are distinct.

\begin{corollary}\label{cor:finequiv}
Let $n$ be a multiple of $3$ and let $\gamma,\xi \in \fqn\setminus \F_{q^3}$. Then $V_{x^q,\gamma}$ and $V_{\mathrm{Tr}_{\F_{q^3}/\fq}(x),\xi}$ are inequivalent.
\end{corollary}

In terms of codes we obtain a classification result for one-orbit cyclic subspace codes of dimension three when $3 \mid n$ and $\delta_3(S)=2$.

\begin{corollary}
Let $C=\mathrm{Orb}(S)$ be a one-orbit cyclic subspace code with dimension three in $\fqn$. Suppose that $3 \mid n$ and $\delta_3( S)=2$. 
\begin{itemize}
    \item If $n>6$ then $d(C)=4$ and $C$ is equivalent either to $\mathrm{Orb}(V_{x^q,\gamma})$ or to $\mathrm{Orb}(V_{\mathrm{Tr}_{\F_{q^3}/\fq}(x),\gamma})$, for some $\gamma \in \fqn\setminus\F_{q^3}$.
    \item If $n=6$ and $d(C)=4$ then $C$ is equivalent either to $\mathrm{Orb}(V_{x^q,\gamma})$ with $\N_{\F_{q^3}/\fq}(\gamma)\ne 1$ or to $\mathrm{Orb}(V_{\mathrm{Tr}_{\F_{q^3}/\fq}(x),\gamma})$ with $\mathrm{Tr}_{\F_{q^6}/\fq}(\gamma)\ne -2$.
    \item If $n=6$ and $d(C)=2$ then $C$ is equivalent either to $\mathrm{Orb}(V_{x^q,\gamma})$ with $\N_{\F_{q^3}/\fq}(\gamma)= 1$ or to $\mathrm{Orb}(V_{\mathrm{Tr}_{\F_{q^3}/\fq}(x),\gamma})$ with $\mathrm{Tr}_{\F_{q^6}/\fq}(\gamma)= -2$.
\end{itemize}
\end{corollary}

In \cite[Example 1]{Etzion} (see also \cite[Example 3.14]{Heidenew} and \cite[Example 18]{Troha}) Etzion and Vardy provide an explicit example of optimum-distance code when $n=6$, $k=3$ and $q=2$.
Under this assumption, the construction $V_{x^q,\gamma}$ of Roth, Raviv and Tamo \cite{Roth} does not provide any example of optimum-distance codes. Therefore, by the above corollary, we know that \cite[Example 1]{Etzion} arises from $V_{\mathrm{Tr}_{\F_{q^3}/\fq}(x),\gamma}$, for some $\gamma \in \F_{2^6}\setminus\F_{2^3}$ with $\mathrm{Tr}_{\F_{2^6}/\F_2}(\gamma)=1$.


\section*{Acknowledgements}

The research was supported by the project ``COMBINE'' of the University of Campania ``Luigi Vanvitelli'' and was partially supported by the Italian National Group for Algebraic and Geometric Structures and their Applications (GNSAGA - INdAM).
This research was also supported by Bando Galileo 2024 – G24-216 and by the project ``The combinatorics of minimal codes and security aspects'', Bando Cassini.

\newpage \clearpage

\begin{appendices}
\section{Some proofs}

Because of lack of space, here we prove some of the stated results.

\emph{Proof of Theorem \ref{thm:polcaseequiv}}\\
Suppose that $T$ and $S$ are equivalent under the action of $(\xi,\sigma)\in \fqn^*\rtimes \mathrm{Aut}(\F_{q^n})$, then $T=\xi S^{\sigma}$, i.e.
\[
\langle 1,\mu\,\mu^2\rangle_{\fq}=\xi\langle 1,\lmb^{\sigma},\lmb^{2\sigma}\rangle_{\fq}.
\]
Then
\begin{equation}\label{eq:polsystsigmalambda}
\begin{cases}
    \xi=p_0(\mu)\\
    \xi\lmb^{\sigma}=p_1(\mu)\\
    \xi\lmb^{2\sigma}=p_2(\mu)
\end{cases}
\end{equation}
where $p_i(x)\in\fq[x]$ and $\deg_{\fq}(p_i(x))\leqslant 2$ for any $i\in\lbrace 0,1,2\rbrace$. Without loss of generality, up to change $\xi$, we may assume that $\gcd(p_0(x),p_1(x),p_2(x))=1$. From \eqref{eq:polsystsigmalambda} we get
\begin{equation}
\begin{cases}
\label{sist1}
    \lmb^{\sigma}=\frac{p_1(\mu)}{p_0(\mu)}\\
    \lmb^{2\sigma}=\frac{p_2(\mu)}{p_0(\mu)}.
\end{cases}
\end{equation}
This implies that
\begin{equation}
\label{eq:p1mup2mu}
p_1^2(\mu)=p_2(\mu)p_0(\mu),
\end{equation}
and since $\dim_{\fq}(\fq(\mu))\geqslant 5$, Equation \eqref{eq:p1mup2mu} implies the following polynomial identity
\begin{equation}
\label{eq1 p1}
p_1^2(x)=p_2(x)p_0(x).
\end{equation}
Since $\lmb^{\sigma}\notin\fq$, we have that $p_1(x)$ and $p_0(x)$ are not $\fq$-proportional and the same holds for $p_1(x)$ and $p_2(x)$. Therefore, by Equation \eqref{eq1 p1} we have that $p_1(x)$ is reducible over $\fq$, i.e.
\[
p_1(x)=t(x)s(x)
\]
where $t(x),s(x)\in\fq[x]$ and $\deg(t(x))=\deg(s(x))=1$. Then, since $\gcd(p_0(x),p_1(x),p_2(x))=1$, by \eqref{eq1 p1} if $t(x)\mid p_2(x)$ then  $t^2(x)|p_2(x)$ and $s^2(x)|p_0(x)$. Thus
$p_0(x)=\alpha s^2(x)
\text{ and }
p_2(x)=\beta t^2(x)$
where $\alpha,\beta\in\fq$ and $\alpha\beta=1$.\\
Then by \eqref{sist1} it follows that
\[
\lmb^{\sigma}=\frac{t(\mu)s(\mu)}{\alpha s^2(\mu)}=\frac{t(\mu)}{\alpha s(\mu)}=\frac{\alpha_0+\alpha_1\mu}{\beta_0+\beta_1\mu}.
\]
for some $\alpha_0,\alpha_1,\beta_0,\beta_1\in\fq$ with $(\alpha_1,\beta_1)\neq (0,0)$. We obtain the same condition in the case in which $s(x)|p_2(x)$.\\
Conversely, suppose that 
\[
\lmb^{\sigma}=\frac{\alpha_0+\alpha_1\mu}{\beta_0+\beta_1\mu}=\frac{t(\mu)}{s(\mu)}.
\]
for some $\alpha_0,\alpha_1,\beta_0,\beta_1\in\fq$ with $(\alpha_1,\beta_1)\neq(0,0)$. Let $\xi=s^2(\mu)$, then
\begin{equation}
    \begin{aligned}
        \xi S^{\sigma}&=s^2(\mu)\left\langle 1, \frac{t(\mu)}{s(\mu)}, \frac{t^2(\mu)}{s^2(\mu)}\right\rangle_{\fq}\subseteq\langle 1,\mu,\mu^2\rangle_{\fq}=T,
    \end{aligned}
\end{equation}
and, since $\dim_{\fq}(\xi S^{\sigma})=\dim_{\fq}(S)=\dim_{\fq}(T)=3$, we get 
\[
\xi S^{\sigma}=\xi\langle 1,\lmb^{\sigma},\lmb^{2\sigma}\rangle_{\fq}=\langle 1,\mu,\mu^2\rangle_{\fq}=T,
\]
i.e. $S$ and $T$ are equivalent.
\qed

\medskip

\emph{Proof of Theorem \ref{thm:inequiv123}}\\
Suppose that $T$ and $S$ are equivalent under the action of $(\xi,\sigma)$, i.e. $T=\xi S^{\sigma}$. 
Then $\mathbb{F}_{q^2}+\xi\mathbb{F}_{q^2}\subseteq T$ and hence $\mathbb{F}_{q^2}=\xi\mathbb{F}_{q^2}$, i.e. $\xi\in\mathbb{F}_{q^2}$. Moreover we get
\[
T=\mathbb{F}_{q^2}+\langle \eta\rangle_{\fq}=\mathbb{F}_{q^2}+\xi\langle\mu^{\sigma}\rangle_{\fq}
\]
and so
\[
\eta=a+\xi\mu^{\sigma}b
\]
where $a\in\mathbb{F}_{q^2}$ and $b\in\fq$.\\
Conversely, suppose that $\eta=a+\xi\mu^{\sigma}b$ with $a,\xi\in\mathbb{F}_{q^2}$ and $b\in\fq^*$. Then 
\[
T=\mathbb{F}_{q^2}+\langle \eta\rangle_{\fq}= \xi\mathbb{F}_{q^2}+\xi\langle \mu^{\sigma}\rangle_{\fq}=\xi S^{\sigma}.
\]
\qed

\medskip

\emph{Proof of Proposition \ref{prop:1polstep}}\\
Suppose by contradiction that $S\cap\xi\mathbb{F}_{q^3}\neq\lbrace 0\rbrace$ for every $\xi\in\fqn^*$ such that $\xi\mathbb{F}_{q^3}\subseteq\langle S\rangle_{\F_{q^3}}$. Then, since 
\[
\langle S\rangle_{\F_{q^3}}\setminus\lbrace 0\rbrace=\displaystyle\bigcup_{\xi\F_{q^3}\subseteq\langle S\rangle_{\F_{q^3}}}\xi\F_{q^3}^*,
\]
this implies that $S^*=S\setminus\lbrace 0\rbrace$ contains at least as many elements as the multiplicative cosets of $\mathbb{F}_{q^3}$ contained in $\langle S\rangle_{\F_{q^3}}$, i.e. since $\delta_2(S)=2$, we have that
\[
q^3-1=|S^*|\geqslant \frac{q^6-1}{q^3-1}= q^3+1
\]
which gives a contradiction. 
\qed

\medskip

\emph{Proof of Proposition \ref{prop:SdefinesFlinearmap} }\\
Let consider $x,y\in S$ such that $x\neq y$. Then there exist $u,u',v,v'\in\mathbb{F}_{q^3}$ such that 
\[
x=\lmb u+\rho u'\text{ and }y=\lmb v+\rho v'.
\]
Let observe that $u'\neq v'$ because otherwise $x-y=\lmb(u-v)\in S\cap\lmb\mathbb{F}_{q^3}$ and this implies $x=y$, a contradiction. Therefore 
 \[
\begin{aligned}
f\colon \mathbb{F}_{q^3}&\to \mathbb{F}_{q^3}\\
u&\mapsto u': \lmb u+\rho u'\in S
\end{aligned}
    \] 
is well defined. Also, since $S=\lbrace \rho u+\lmb f(u):u\in\mathbb{F}_{q^3}\rbrace$ is an $\fq$-subspace, $f$ is an $\fq$-linear map, i.e. $f \in \mathcal{L}_{3,q}$.
\qed

\medskip

The following is a well-known result in linear sets theory (see e.g. \cite[Proposition 2.2]{NPSZcompl}), but we reformulate it in a more algebraic flavour including also a short proof. This will allow us to prove Corollary \ref{cor:finequiv}.

\begin{proposition}\label{prop:weightinvariace}
Let $V_1=V_{U_1,\gamma}$ and $V_2=V_{U_2,\xi}$ be two $k$-dimensional equivalent $\fq$-subspaces in $\F_{q^n}$. For every $i \in \{0,1,\ldots,k\}$ and $j \in \{1,2\}$, define
\[ N_i(U_j)=|\{ \langle v \rangle_{\fqk} \colon v\in\fqk^2\setminus\lbrace 0\rbrace \text{ and } \dim_{\fq}(U_i\cap \langle v \rangle_{\fqk})=i \}|. \]
Then $N_i(U_1)=N_i(U_2)$ for any $i$.
\end{proposition}
\begin{proof}
By Theorem \ref{thm:equiv} it follows that $U_1$ and $U_2$ are $\mathrm{\Gamma L}(2,q^k)$-equivalent, then there exist a matrix $A \in \mathrm{GL}(2,q^k)$ and an automorphism $\rho \in \mathrm{Aut}(\fqk)$ such that
\[ U_1^{\rho}A=U_2. \]
Let $v \in \fqk^2 \setminus\{(0,0)\}$.
Then $\dim_{\fq}(U_1 \cap \langle v \rangle_{\fqk})=i$ if and only if $i=\dim_{\fq}(U_1^{\rho}A \cap \langle v^{\rho}A \rangle_{\fqk})=\dim_{\fq}(U_2 \cap \langle v^{\rho}A \rangle_{\fqk})$. This means that $N_i(U_1)=N_i(U_2)$.
\end{proof}

\medskip

\emph{Proof of Corollary  \ref{cor:finequiv}}\\
By contradiction, suppose that $V_{x^q,\gamma}$ and $V_{\mathrm{Tr}_{\F_{q^3}/\fq}(x),\gamma}$ are equivalent, then by Theorem \ref{thm:equiv} the $\fq$-subspaces
\[ U_{x^q}=\{ (u,u^q) \colon u \in \F_{q^3} \} \]
and 
\[ U_{\mathrm{Tr}_{\F_{q^3}/\fq}(x)}=\{ (u,\mathrm{Tr}_{\F_{q^3}/\fq}(u)) \colon u \in \F_{q^3} \} \]
are $\mathrm{\Gamma L}(2,q^3)$-equivalent.
Note that $U_{x^q}$ is scattered, that is for any $v \in \F_{q^3}^2$
\[ \dim_{\fq}(U_{x^q} \cap \langle v \rangle_{\F_{q^3}})\leq 1, \]
whereas 
\[ \dim_{\fq}(U_{\mathrm{Tr}_{\F_{q^3}/\fq}(x)} \cap \langle (1,0) \rangle_{\F_{q^3}})=\dim_{\fq}(\ker(\mathrm{Tr}_{\F_{q^3}/\fq}(x)))=2. \]
Therefore, Proposition \ref{prop:weightinvariace} yields a contradiction.
\qed

\end{appendices}

\end{document}